\pdfoutput=1
\documentclass[a4paper,12pt]{amsart}
\usepackage{amsmath,amssymb,cite}

\newcommand{\R}{{\mathbb R}}
\newcommand{\Z}{{\mathbb Z}}
\newcommand{\Cpx}{{\mathbb C}}
\newcommand{\V}{\mathcal V}

\newcommand{\image}{\mathrm{Im}}
\newcommand{\diag}{\mathrm{diag}}
\newcommand{\tr}{\mathrm{trace}}
\newcommand{\cl}{\mathrm{cl}}

\newtheorem{Theorem}{Theorem}[section]

\newtheorem{Proposition}[Theorem]{Proposition}

\theoremstyle{definition}
\newtheorem{Definition}[Theorem]{Definition}
\newtheorem{Example}[Theorem]{Example}

\theoremstyle{remark}
\newtheorem{Remark}[Theorem]{Remark}

\title[Nakayama topos and Bayesian quantum computing]{Nakayama's reduction of quantum topos and Bayesian quantum computing}

\author{Atsuhide Mori}
\address{Department of Mathematics, Osaka Dental University}
\email{mori-a@cc.osaka-dent.ac.jp}


\begin{document}

\begin{abstract}
Nakayama found the way to limit the contexts in the D\"oring-Isham topos quantum theory and unified such limitation and the definition of probability as a single choice of topology. We explain how this unification supports the Bayesian perspective where the probability changes as data is obtained. We describe this change as a local shift of the applicability of a predictive theory and apply it to quantum computing.  
\end{abstract}

\maketitle

\section{Introduction}
Quantum theory starts by setting a Hilbert space, but then soon the mathematical interest turns to the non-commutative algebra itself, say $A$, realizable as a system of operators on such space. This is also the physics motivation for considering the Heisenberg picture. The topos quantum theory of D\"oring and Isham \cite{DI1,DI2,DI3,DI4,D5,DI6} deals with the poset $\V_A$ of commutative von Neumann subalgebras of $A$ to understand how $A$ allows quantum contextuality including Bell non-locality. 

A topos theory is a parametric version of set theory. The poset $\V_A$ serves as the parameter space for the quantum topos. This meta-theory addresses a system without exterior including observers, 
but any proper theory must have its limit of applicability. 
Nakayama \cite{Nakayama1,Nakayama2,Nakayama3} showed that such a limit appears in a coarse-graining of the topos theory. 
Namely, he set up an appropriate topology on the quantum topos to a given subset $S$ of $A$ so that $S$ becomes barely available in the corresponding coarse-graining of the topos. 
Furthermore, he mixed it with the other topology on the topos with base poset $\V_A\times[0,1]$ which defines the probability in the D\"oring-Isham theory, meaning that {\em selecting contexts for thought and thinking in terms of certainty are integrated in reducing the meta-theory into a predictive theory}. 

In this article, we incorporate a Bayesian perspective into Nakayama's approach: {\em Any probability is ever-changing, and therefore the selection of contexts is also a subject to change}. We describe how a portion of the selected contexts becomes a data and excluded from the next selection. We are not attempting to subjectivize the notion of probability like the QBism does (see \cite{Critique} for critiques), but rather propose a possible mechanism by which variables become constants as the applicability of the theory locally shifts. In a certain extent, Nakayama's original work naturally includes this perspective, hence we supplement this point in the overview of his theory in \S 3.1. Following this in \S 3.2, we propose the Bayesian perspective as a framework in the Heisenberg picture. In practical level, we apply the theory to quantum computing in \S 4. A brief preliminary on quantum computing is provided in \S 2 to share the notation with readers unfamiliar with this topic. 

\section{Preliminary on quantum computing}
For details on this section, we refer the readers to a textbook on quantum computing such as \cite{NC}. Let $P_1$ be the original Pauli group generated by the Pauli matrices 
\[
\sigma_0=\begin{pmatrix}1&0\\0&1\end{pmatrix}, \quad
\sigma_1=\begin{pmatrix}0&1\\1&0\end{pmatrix}, \quad
\sigma_2=\begin{pmatrix}0&-i\\i&0\end{pmatrix}, \quad
\sigma_3=\begin{pmatrix}1&0\\0&-1\end{pmatrix}.
\] 
They are all unitary and Hermitian (hence their squares are equal to $\sigma_0$ which is the unit) and satisfy the formulas
\[
\sigma_1\sigma_2=-\sigma_2\sigma_1,\quad 
\sigma_2\sigma_3=-\sigma_3\sigma_2,\quad
\sigma_3\sigma_1=-\sigma_1\sigma_3,\quad
\sigma_1\sigma_2\sigma_3=i\sigma_0.
\]
Thus, each element of $P_1$ is presented by  
$i^k{\sigma_3}^p{\sigma_1}^q$ ($k\in \Z_4, p,q\in\Z_2$). 
The (extended) Pauli group $P_n$ is the $n$-fold tensor product of $P_1$, which consists of $2^{2n+2}$ matrices of size $2^n \times 2^n$. Here we treat the tensor product of two matrices $a$ and $b$ with arbitrary sizes as the Kronecker product $a\otimes b=\begin{pmatrix}a_{11}&\cdots&a_{1c}\\\vdots&&\vdots\\ a_{r1}&\cdots&a_{rc}\end{pmatrix}\otimes b=\begin{pmatrix}a_{11}b&\cdots&a_{1c}b\\\vdots&&\vdots\\ a_{r1}b&\cdots&a_{rc}b\end{pmatrix}$, i.e., the block matrix consisting of the scalar multiples of the matrix $b$ by the entries of $a$. Note that the matrix product $(a\otimes b)(a'\otimes b')$ can be calculated as $(aa')\otimes (bb')$ provided that both $aa'$ and $bb'$ are defined. The projectivization $P_n/\Z_4$ is of order $2^{2n}$. The quantum Clifford group $C_n$ is the normalizer of $P_n$ in the unitary group $\mathrm{U}(2^n)$, i.e., 
\[
C_n:=N(P_n)=\{x\in \mathrm{U}(2^n)\mid x P_n x^*=P_n\}(\rhd P_n)
\]
and $C_n/\mathrm{U}(1)$ its projectivization. 
The quotient $(C_n/\mathrm{U}(1))/(P_n/\Z_4)$ is isomorphic to the symplectic group $\mathrm{Sp}(2n,\Z_2)$ since each element of $P_n/\Z_4$ can be written in the form  $\sigma_{p,q}:={\sigma_3}^{p_1}{\sigma_1}^{q_1}\otimes\cdots\otimes{\sigma_3}^{p_n}{\sigma_1}^{q_n}$ ($p,q \in {\Z_2}^n$) where the symplectic form over $\Z_2$ is 
\[
\Omega_{p,q,r,s}=\begin{pmatrix}p^T&q^T\end{pmatrix}\begin{pmatrix}0&\mathrm{id}\\\mathrm{id}&0\end{pmatrix}\begin{pmatrix}r\\s \end{pmatrix}(\Leftrightarrow\sigma_{p,q}\sigma_{r,s}=(-1)^{\Omega_{p,q,r,s}}\sigma_{r,s}\sigma_{p,q}).
\]
Since there are $(2^{2n}-1)\cdot 2^{2n-1}$ choices for the first pair of symplectic basis, we have the cardinality formula
\[
|C_n/\mathrm{U}(1)|=|P_n/\Z_4|\cdot|\mathrm{Sp}(2n,\Z_2)|=\displaystyle 2^{2n+n^2}\prod_{j=1}^n(2^{2j}-1),
\]
e.g., $|C_1/\mathrm{U}(1)|=24$, $|C_2/\mathrm{U}(1)|=24\cdot4\cdot2^3\cdot(2^4-1)=11520$. 

In quantum computing, quantum gates and their series circuits are represented by unitary matrices and their matrix products, namely, an $n$-qubit gate is an element of $\mathrm{U}(2^n)$. We consider the tensor product of an $m$-qubit gate and an $n$-qubit gate as an $(m+n)$-qubit gate. We can embed a non-trivial Pauli matrix $\sigma_j$ ($j=1,2,3$) as an $n$-qubit gate ${\sigma_0}^{\otimes(k-1)}\otimes\sigma_j\otimes {\sigma_0}^{\otimes(n-k)}$ for $k=1,\dots, n$. Any such gate is called Pauli gate. Similarly, we define Hadamard gate by embedding the H-matrix 
\[
h=\frac{\sigma_1+\sigma_3}{\sqrt{2}}
=\frac{1}{\sqrt{2}}\begin{pmatrix}1&1\\1&-1\end{pmatrix}=h^*\in C_1\subseteq \mathrm{U}(2),
\]
to $C_n$ in $n$ ways, phase gate by embedding the S-matrix
\[
s=\frac{\sigma_0+\sigma_3}{2}+i\frac{\sigma_0-\sigma_3}{2}=\begin{pmatrix}1&0\\0&i\end{pmatrix}=\sigma_3 s^*\in C_1\subseteq \mathrm{U}(2),
\]
in $n$ ways and CNOT gate, provided that $n>1$, by embedding the C-matrix
\[
c=\frac{1}{2}\sum_{u,v\in \Z_2}(-1)^{uv}{\sigma_3}^u \otimes {\sigma_1}^v=\begin{pmatrix}1&0&0&0\\0&1&0&0\\0&0&0&1\\0&0&1&0\end{pmatrix}=c^*\in C_2\subseteq \mathrm{U}(4)
\]
in $n-1$ ways. The following calculation shows that these specific gates belong to $C_n$. $h\sigma_1=\sigma_3 h$, $h\sigma_2=-\sigma_2h$, $s\sigma_1=\sigma_2 s$, $s\sigma_2=-\sigma_1 s$, $s\sigma_3=\sigma_3s$, 
$c(\sigma_1\otimes \sigma_3)=-(\sigma_2\otimes \sigma_2)c$, and $c(\sigma_j\otimes \sigma_k)=(\sigma_l\otimes\sigma_m)c$ for 
$(j,k;l,m)=(0,0;0,0)$, $(0,1;0,1)$, $(0,2;3,2)$, $(0,3;3,3)$, $(1,0;1,1)$, 
$(1,2;2,3)$, $(2,0;2,1)$, $(3,0;3,0)$, $(3,1;3,1)$.\par
The S-matrix can also be considered as the $\displaystyle \frac{\pi}{4}$-matrix $\displaystyle r\left(\frac{\pi}{2}\right)$, which is an example of the following $\displaystyle\frac{\theta}{2}$-matrix $r(\theta)$. 
\[
r(\theta)=\begin{pmatrix}1&0\\0&\exp(i\theta)\end{pmatrix}=\exp(i \theta/2)\begin{pmatrix}\exp(-i\theta/2)&0\\0&\exp(i\theta/2)\end{pmatrix}
\]
Now we summarize the fundamental results in quantum computation.
\begin{Theorem}\begin{enumerate}
\item[i)](Gottesman-Knill \cite{Got})\quad
The group $C_n/\mathrm{U}(1)$ is the quotient of the group generated by Hadamard, phase, and CNOT gates under the $\Z_8$-action on the phase factor. Clifford circuits with computational-basis inputs and Pauli measurements admit efficient classical simulation. 
\item[ii)](Boykin et al. \cite{Boykin})\quad 
Adding the square root of phase gate defined by the T-matrix 
\[
t=\begin{pmatrix}1&0\\0&\exp(i\pi/4)\end{pmatrix}(=r(\pi/4):~\textrm{the $\displaystyle\frac{\pi}{8}$-matrix}~)
\]
to the above generators, we can approximate any element of $\mathrm{U}(2^n)$ arbitrarily well up to a global phase.  
\end{enumerate}\label{qubits}
\end{Theorem}

\begin{Remark} A vector of the form $x\otimes y=\begin{pmatrix}x_1\\x_2\end{pmatrix}\otimes\begin{pmatrix}y_1\\y_2\end{pmatrix}\in \Cpx^4$ is said to be separable. The matrix $c$ sends a separable vector $x\otimes y$ to
\begin{align*}
c(x\otimes y)&=\begin{pmatrix}x_1y_1\\x_1y_2\\x_2y_2\\x_2y_1\end{pmatrix}
=\begin{pmatrix}x_1\\0\end{pmatrix}\otimes\begin{pmatrix}y_1\\y_2\end{pmatrix}+
\begin{pmatrix}0\\x_2\end{pmatrix}\otimes\begin{pmatrix}y_2\\y_1\end{pmatrix}\\
&=\frac{1}{2}\begin{pmatrix}x_1\\x_2\end{pmatrix}\otimes\begin{pmatrix}y_1+y_2\\y_1+y_2\end{pmatrix}
+\frac{1}{2}\begin{pmatrix}x_1\\-x_2\end{pmatrix}\otimes\begin{pmatrix}y_1-y_2\\-(y_1-y_2)\end{pmatrix}
\end{align*}
which is non-separable i.e., {\em entangled} if $x_1x_2(y_1+y_2)(y_1-y_2)\neq 0$. 
In classical logic, there are no entanglements; $\begin{pmatrix}1\\0\end{pmatrix}$ represents the truth value $0$; and $\begin{pmatrix}0\\1\end{pmatrix}$ the other truth value $1$. Then $c$ works as a logic gate that negates $y$ only when $x$ is true, hence Controlled NOT gate. 
To swap the components $x$ and $y$, one can use the matrix 
\[
\mathrm{swap}:=c(h\otimes h)c(h\otimes h)c=\begin{pmatrix}1&0&0&0\\0&0&1&0\\0&1&0&0\\0&0&0&1\end{pmatrix}:
x\otimes y\mapsto y\otimes x.
\]
This is why ${\sigma_0}^{\otimes(k-1)}\otimes c\otimes {\sigma_0}^{\otimes(n-k-1)}$ ($k=1,\dots, n-1$) are sufficient. 
\end{Remark}

\section{Topos quantum theory}
\subsection{The D\"oring-Isham-Nakayama theory}
Let $A=B(H)$ be the algebra of bounded linear operators on a Hilbert space $H$,
with the Hermitian adjoint $^*$ as its involution.
For any subset $S\subseteq A$, its commutant $S'$ consists of the operators in $B(H)$ commuting with every element of $S$. All commutants below are taken in $B(H)$.
By the von Neumann double commutant theorem, $V=(S\cup S^*)''$ is the smallest weakly closed unital $^*$-subalgebra of $B(H)$ containing $S$, namely the von Neumann algebra generated by $S$. If $S\subseteq M\subseteq B(H)$ for a unital von Neumann subalgebra $M$, then $V\subseteq M$.
There is the absolute version of von Neumann algebra called $W^*$-algebra, but any $W^*$-algebra is realizable as above.  
In this article, we assume $H$ to be separable for considering density operators and probabilities. 
We refer to \cite{Takesaki} for the mathematical theory of operator algebras. 

The topos quantum theory of D\"oring and Isham \cite{D5,DI6} deals with the poset $(\V_A,\subseteq)$ of all {\em commutative} von Neumann subalgebras of $A$ with the same unit $\mathrm{id}_H$ under inclusions, which is considered as the whole of contexts. 
A presheaf on a poset $(\V,\subseteq)$ is a pair $X=(X_{\mathrm{obj}},X_{\mathrm{mor}})$ of assignments of a set $X_{\mathrm{obj}}(V)$ for each element $V\in \V$ and a map $X_{\mathrm{mor}}(U\subseteq V): X_{\mathrm{obj}}(V)\to X_{\mathrm{obj}}(U)$ for each inclusion relation $U\subseteq V$ which satisfies the following rules:
\begin{align*}
 X_{\mathrm{mor}}(V\subseteq V)
&=\mathrm{Id}_{X_{\mathrm{obj}}(V)}\quad(\forall V)\\
 X_{\mathrm{mor}}(T\subseteq U)\circ X_{\mathrm{mor}}(U\subseteq V)
&=X_{\mathrm{mor}}(T\subseteq V)\quad(\forall T\subseteq U\subseteq V) 
\end{align*}
We omit the subscriptions obj and mor unless afraid of confusion. The set of all presheaves on a poset $\V$ defines a topos, but the topos quantum theory deals only with special presheaves. So, in this article, we treat any subset as such in naive set theory, not using the abstract description as a subobject in the category of sets.

A point or a global presheaf section of $X$ is an element of 
\[
\Gamma X=\left\{f:\V\to \bigsqcup_{V\in \V}X(V) \mid 
\begin{array}{l}
f(V)\in X(V) ~(\forall V) \quad\textrm{and} \\
f(U)=X(U\subseteq V)(f(V)) ~(\forall U\subseteq V) 
\end{array}
\right\}
\]
which is possibly empty. 
A point in classical topology can be thought as a constant point of a constant presheaf over the poset of open sets. 
In topos theory, we can not use points to define a topology. Instead, we use a closure operation for subpreshaves as a topology. 
Here a subpresheaf $Y\subseteq X$ of a presheaf $X$ on $\V$ is a presheaf on $\V$ which satisfies $Y(V)\subseteq X(V)$ for $\forall V\in \V$ and $Y(U\subseteq V)=X(U\subseteq V)|_{Y(V)}$ ($\Rightarrow X(U\subseteq V)(Y(V))\subseteq Y(U)$) for $\forall U\subseteq V$, and the sets of such presheaves is denoted by $\mathrm{Sub}(X)$. 

Any subset $S\subseteq A$ defines the idempotent map $\flat=\flat_S:\V_A\to\V_A$ by 
\[
\flat_S(V)=\bigl((S\cup S^*)\cap V\bigr)''.
\] 
We define the $\flat$-induced closure of a subpresheaf $Y\in \mathrm{Sub}(X)$ as the maximal presheaf $\cl=\cl_{X,\flat} Y\in \mathrm{Sub}(X)$ satisfying
\[
X(\flat (V)\subseteq V)(\cl_{X,\flat} Y(V))\subseteq Y(\flat V)\quad(\forall V\in \V_A).
\]
The fibers $\{X(V)\}_{V\in\flat(\V_A)}$ on the image of $\flat$ form a presheaf, which we denote by $X_\flat$ and call the $\flat$-reduction of $X$.  The $\flat$-induced topology on the reduction $X_\flat$ is called the discrete topology since $\cl=\textrm{id}$. 
As a whole, the $\flat$-induced topology merely serves as a foundation; we further consider other geometry and topology built upon it. 

We use the following results in quantum theory: The set $\mathrm{Hom}(V,\Cpx)$ of unital algebra homomorphisms from $V\in \V_A$ to $\Cpx$ carries the weak topology as the Stone dual to the Boolean algebra of the projections in $V$. It is a compact Hausdorff space such that the closure of any open set is open. Here the dual of a projection $\pi$ in $V$ is the clopen subset of $\mathrm{Hom}(V,\Cpx)$ consisting of homomorphisms $\lambda$ taking the value $1$ for the projection, i.e., $\lambda(\pi)=1$. Further, each inclusion $U\subseteq V$ induces an open continuous surjection $\mathrm{Hom}(V,\Cpx)\to \mathrm{Hom}(U,\Cpx)$. 

\begin{Definition}
The spectral presheaf $\Sigma_A$ on $\V_A$ is defined by $\Sigma(V)=\mathrm{Hom}(V,\Cpx)$ and $\Sigma_A(U\subseteq V):\lambda\mapsto \lambda|_U$. Let $r_V=\Sigma_A(\flat V\subseteq V)$ denote restriction. A subpresheaf $X\subseteq\Sigma_A$ is called a $\flat$-proposition if $X(\flat V)$ is clopen and \mbox{$X(V)=r_V^{-1}(X(\flat V))$} for every $V\in\V_A$. In the case where $\flat$ is the identity, a $\flat$-proposition is called a proposition. Let $\flat^{-1}\mathrm{Prop}((\Sigma_A)_\flat)$ denote the set of $\flat$-propositions and $\mathrm{Prop}(\Sigma_A)$ that of propositions. 
\end{Definition}

For a proposition $X$, the dual of $X(\flat V)$ is a projection in $\flat V$, and can also be considered as a projection in $V$. As the restriction $\Sigma_A(\flat V\subseteq V)$ is continuous, the dual of this projection in $V$ is again a clopen subset, which is by definition the $\flat$-induced closure $\cl X(V)$. Note that it is different from the classical closure of the clopen subset $X(V)$ under the weak topology if it is larger than $X(V)$ and this is the case when $X$ is not a $\flat$-proposition. In other word, a $\flat$-proposition is the fattening of a proposition up to its $\flat$-induced closure. Further, the restriction is a homeomorphism if it is injective. Indeed we have   

\begin{Theorem}
\begin{enumerate}

\item[(i)] (Nakayama \cite{Nakayama1,Nakayama2,Nakayama3}) Each fiber of the reduction of a $\flat$-proposition to the image $\flat(\V_A)$ is a clopen subset of a fiber of $\Sigma_A$. The $\flat$-proposition is a trivial extension of this reduction meaning that, for each $V\in \V_A$, the reduction of the $\flat$-projection to the preimage $\flat^{-1}(\flat V)$ is a constant presheaf or its redundance by adding extra elements to each location in each fiber. 

\item[(ii)] (Bayesian version) Given a set $T$ with $S\subseteq T\subseteq A$, we put $\sharp=\flat_T$ and take the reduction $\Sigma=(\Sigma_A)_\sharp$ of the spectral presheaf $\Sigma_A$ to the poset $\V=\sharp(\V_A)$. Then, we can replace all $\V_A$ and $\Sigma_A$ in the above arguments with $\V$ and $\Sigma$ respectively and adjust the domain of $\flat$ to paraphrase the above statement $\mathrm{(i)}$.  

\end{enumerate}
\end{Theorem}

\begin{proof} 
For the proof of (ii), we apply (i) by replacing $\flat$ with $\sharp$ and compare the result with (i). It can also be achieved by changing $\Sigma_A$ with $\Sigma$ (over any subposet) and following the original proof of (i). 
\end{proof}

We notice that (ii) implies (i) when $T=A$. Hereafter, we drop the subscript $A$ to adopt the setting for the more general case (ii). 

\begin{Definition} Let $\Pi=\{\pi \in A \mid \pi^2=\pi=\pi^*\}$ be the set of projections. For each context $V\in \V$, 
we define the $\flat$-outer estimation 
\[
\mathrm{O}_V: \Pi\to \Pi\cap \flat V: \pi \mapsto \inf\{a\in \Pi\cap \flat V\mid \image a \supseteq\image\pi\}
\]
at $V$. For a proposition $X$, let $\mu(X,V)$ be the unique projection $\pi\in\Pi\cap V$ such that
\[
X(V)=\{\lambda\in\Sigma(V)\mid\lambda(\pi)=1\}.
\]
If $X$ is a $\flat$-proposition, then $\mu(X,V)\in\flat V$; we call this restriction of the duality map the $\flat$-man.
The $\flat$-daseinization\footnote{Heidegger's Dasein is a practical approach to some project in the world that enables a particular choice, while his Das Man is still an unresolved, as-is variety of cases. The term ``man'' does not appear in topos quantum theory, but its inclusion here fits the notion of contextual probability. We notice that the term daseinization is quite misleading; Das Man is Dasein unaware of resolve and thus the shift that actualize authentic Dazein should be termed resolution or resolve. } is the map $\delta: \Pi\to \flat^{-1}\mathrm{Prop}(\Sigma_\flat)$ defined by  
\[
\mathrm{O}_V(\pi)=\mu(\delta(\pi),V).
\]
If $\flat$ is the identity, these maps are called the outer estimation, the man, and the daseinization, respectively. 
\end{Definition}

\begin{Remark} 
The infimum of any family of projections is a projection. The intersection or infimum of a family of ($\flat$-)propositions is a proposition which is given fiberwise by the interior of the set-theoretic intersection, while the union or supremum of infinite family is the fiberwise closure of the union under the weak topology. The intersection and union of propositions form a Heyting algebra, i.e., a bounded distributive lattice with the intuitionistic implication $Z=(X\Rightarrow_{\mathrm{intuit}} Y)$ which is the maximal proposition $Z$ satisfying $X\cap Z\subseteq Y$. It is Boolean only when the intuitionistic negation $\neg_{\mathrm{intuit}} X:=(X\Rightarrow \emptyset)$ satisfies the law of excluded middle $\neg_{\mathrm{intuit}} X\cup X=\Sigma$. 
\end{Remark}

\begin{Definition} The contextual probability of a $\flat$-proposition $X$ on a density $\rho\in P_A=\{a\in A \mid a=a^*\geq 0 ~\textrm{ and }~ \tr(a)=1\}$ is 
\[
\mu_\rho(X,V)=\tr(\rho \mu(X,V)).
\]
We use the same formula for an arbitrary proposition $X$.
\end{Definition}
\begin{Remark}
The projection $\mu(X,V)$ is not always a trace class, but the composition $\rho \mu(X,V)$ is. 
For $\pi\in\Pi$ and $V\in\V$, we have $\mu_\rho(\delta(\pi),V)\geq\tr(\rho\pi)$, with equality whenever $\pi\in\flat V$. Thus the Born probability is the minimum over $\V$ if such a context exists. 
\end{Remark} 

\subsection{Bayesian perspective} Following on from above, we apply
\begin{align*}
&S\subseteq T\subseteq A, \quad
\sharp=\flat_T, \quad
\V=\sharp(\V_A), \quad
\Sigma=(\Sigma_A)_\sharp, \quad 
\flat=(\flat_S)|_{\V}, \\
&\mathrm{Prop}(\Sigma)=\{X\subseteq \Sigma \mid X(V): \textrm{clopen }~ (\forall V)\}\supseteq \flat^{-1}\mathrm{Prop}(\Sigma_\flat).
\end{align*}
We propose a Bayesian updating of these items within the Heisenberg picture. 
For a density $\rho\in A$, there exists a unitary operator $u\in A$ diagonalizing $\rho$, i.e., $u^* \rho u$ is diagonal. 
Thus, we shift the choice of $T$ and its subset $S$ by conjugating them with $u^*$ if necessary to assume that the density $\rho$ is diagonal and further the eigenvalues are sorted in monotonically decreasing order from the outset. 

Our first hypothesis is that the system of operators behaves as a geometric figure under the symmetry of the unitary group. This means that the global unitary rotation of $T$ and $S$ is only inferred through the change of probability with respect to the fixed density $\rho$, whereas the inner products between the operators are clear and invariant under the rotation. However, if we permit any rotation, the notion of probability becomes uncontrollable; if we use only rotations commuting with the density, the probability ceases to fluctuate even slightly. 

Instead, we localize the unitary group as follows. We work with the fixed pure-state density $\rho=\diag(1,0,0,\dots)$ and restrict subsequent changes of basis of $H$ to those preserving this form. Particularly, when $H$ is decomposed into a tensor product, the form of the density in each component must be $\diag(1,0,0,\dots)$. Now we take a tensor decomposition $H=H_1\otimes H_2$ with $m=\dim H_2$ satisfying $1<m<\infty$. Even when the set of operators is rotated as a figure, the density $\rho$ remains anchored as 
\[
\rho=\rho_1\otimes\rho_2,\quad \rho_1=\diag(1,0,0,\dots),\quad \rho_2=\diag(1,0,\dots,0).
\]
The relative complement $T\setminus S$, instead of $S$, yields a poset $\overline{\flat(\V)}$ called the complementary poset. We make the following assumptions. 
\begin{enumerate}
\item[(i)] All contexts $V$ in the complementary poset $\overline{\flat(\V)}$ are of the form $V_1\otimes V_2$ and there exists a context with $\dim V_2=m$. 
\item[(ii)] There exists a unitary operator $\widetilde{v}$ of the form $\mathrm{id}_{H_1}\otimes v$ which partially diagonalizes all contexts in the complementary poset $\overline{\flat(\V)}$ as $\widetilde{v}^*(V_1\otimes V_2)\widetilde{v}=V_1\otimes$ (a set of diagonal matrices). 
\item[(iii)] On the probability, any projection $\pi=\pi_1\otimes \pi_2$ in any context of the complementary poset $\overline{\flat(\V)}$ satisfies $\tr(\rho_1\pi_1)=1$. 
\end{enumerate}
These assumptions form our second hypothesis on localizability. 

\begin{Theorem} Under the geometricity and localizability hypotheses, we substitute $\widetilde{v}^* S\widetilde{v}$ and $\widetilde{v}^*T\widetilde{v}$ for $S$ and $T$ respectively to update the setup. Then the updated contextual probability of a proposition in a context of $\overline{\flat(\V)}$ becomes the truth value $0$ or $1$. Considering the choice of the first column of $v$ under a permutation of the columns as a sampling, the updated contextual probability of a $\flat$-proposition is interpreted as a conditional probability. 
\end{Theorem}

\begin{proof} The assumptions (ii) and (iii) imply that any projection in any context of the complementary poset $\overline{\flat(\V)}$ becomes diagonal after the substitution. Then the assumption (i) ensures that the unitary matrix $v\in \mathrm{U}(m)$ is uniquely determined up to permutation and phase-shift of columns. Here each column corresponds to 
a proposition which is true only when that column takes the first place. 
The probability of a proposition depends on the choice of $v$, but only on the first column $v_0$ of $v$ and the top-left $m\times m$-block $\mu_{00}$ of $\mu$ as follows. 
\begin{align*}
&\tr\begin{pmatrix}
\rho_2&0&\cdots\\
0&0&\\
\vdots&&\ddots
\end{pmatrix}
\begin{pmatrix}
v^*&0&\dots\\
0&v^*&\\
\vdots&&\ddots
\end{pmatrix}
\begin{pmatrix}
\mu_{00}&\mu_{01}&\cdots\\
\mu_{10}&\mu_{11}&\\
\vdots&&\ddots
\end{pmatrix}
\begin{pmatrix}
v&0&\cdots\\
0&v&\\
\vdots&&\ddots
\end{pmatrix}\\
&=\tr\begin{pmatrix}
1&0&\cdots\\
0&0&\\
\vdots&&\ddots
\end{pmatrix}
\begin{pmatrix}
v_0^*\\
v_1^*\\
\vdots
\end{pmatrix}
\mu_{00}
\begin{pmatrix}
v_0&v_1&\cdots
\end{pmatrix}=v_0^*\mu_{00}v_0
\end{align*}
Thus, the posterior contextual probability is determined by the data showing which of such propositions become true. Further, the set of those propositions is no longer divisible in a $\flat$-proposition, and the information on $v$ including which is true is absorbed into the internal parameter of the theory and is no longer a subject of prediction. 
\end{proof}
This phenomenon where a proposition with probability neither $0$ nor $1$ becomes either $0$ or $1$ due to local changes would be what has been considered as wave function collapse. In our perspective, it changes a predictable quantity into data by shifting the application range of the theory. Since this change itself obeys the prior probability, the updating process will be a stochastic process once we can program an appropriate sequence of sampling methods. 

\section{Topos theory on quantum computing} 
Hereafter we put $H=\Cpx^{2^n}$ and apply the Bayesian version of the Nakayama theory to quantum computing. The algebra $A$ is the ring of matrices and each context $V$ consists of simultaneously diagonalizable matrices with specific patterns of multiplicity in their eigenvalues. We change the basis as needed, but we always assume that the density takes the form $\rho=\diag(1,0,\dots,0)$. Particularly, for a tensor decomposition $H=H_1\otimes H_2=\Cpx^{2^{n-k}}\otimes \Cpx^{2^k}$, we assume 
\[
\rho=\rho_1\otimes \rho_2=\diag(1,0,\dots,0)\otimes\diag(1,0,\dots,0).
\]
while we can change the role of $\rho_1$ and $\rho_2$ at will. Any subpresheaf $X$ of $\Sigma$ is a proposition since $\Sigma(V)$ is discrete in the weak topology. The $\flat$-induced closure is $\cl X(V)=r_V^{-1}(X(\flat V))$, where $r_V:\Sigma(V)\to\Sigma(\flat V)$ is restriction. Let $u\in \mathrm{U}(2^n)$ be a unitary matrix diagonalizing  $V$. Any matrix in $V$ is then written in the form 
\[
\diag_{u}(x_0, \dots, x_{2^n-1}):=u\diag(x_0, \dots, x_{2^n-1})u^*.
\]
We fix the notation for the spectra as follows. 
\[
\lambda_j^u: \diag_{u}(x_0, \dots, x_{2^n-1})\mapsto x_j \quad (j=0,\dots,2^n-1).
\]
For a subcontext $V$, we represent $X(V)$ by its inverse image under restriction from the spectrum of the maximal diagonal algebra in the chosen basis. 
If $V$ is of dimension $2^n$, it is the complexification of the maximal torus 
\[
T^{2^n}_u:=\{\diag_u(x_0, \dots x_{2^n-1})\mid x_0,\dots,x_{2^n-1} \in \mathrm{U}(1)\}.
\]
The homogeneous space $\mathrm{U}(2^n)/T^{2^n}_{\mathrm{id}}$ is a complete flag manifold.

\subsection{1-qubit case} 
When $n=1$, each maximal torus corresponds to a pair of antipodes on the sphere $\mathrm{U}(2)/T^2_{\mathrm{id}}=\Cpx P^1$. Identifying them, we obtain the real projective plane $\R P^2$, which is the quotient $\mathrm{U}(2)/N(T^2)$ by the normalizer of $T^2$. The complement $U=\mathrm{U}(2)\setminus S^1$ of the circle $S^1=\mathrm{U}(1)\mathrm{id_2}$ of scalar unitary matrices is fibered by open annuli over $\R P^2$. The union of each fiber with $S^1$ is a maximal torus. The fiber containing a non-scalar unitary matrix 
\[
\diag_u(d_0,d_1)=\diag_{u\sigma_1}(d_1,d_0) \in U \quad (d_0, d_1\in \mathrm{U}(1), d_0\neq d_1)
\] 
and the other matrix $\diag_u(d_1,d_0)$ accompanies the context 
\[
V^2_1(u)=\{\diag_u(x_0,x_1)\mid x_0,x_1\in\Cpx\}=V^2_1(u\sigma_1)
\]
of dimension two. The set $\Sigma(V^2_1(u))$ consists of the two spectra 
\[
\lambda_0^u:\diag_u(x_0,x_1)\mapsto x_0,\quad \lambda_1^u:\diag_u(x_0,x_1)\mapsto x_1.
\] 
The annulus bundle has no cross section ($\pi_1 U=\Z$, $\pi_1\R P^2=\Z_2$), 
so we take a double section $D\cong S^2$ in $U$ which meets each fiber transversely at two points and satisfies $\sigma_1 D\sigma_1=D$. 

\begin{Proposition}
Let $T\subseteq D$, with $\diag_v(e_0,e_1)$ its sole representative of the context $V^2_1(v)$, and put $S=T\setminus\{\diag_v(e_0,e_1)\}$. 
We write any matrix in $S$ in the form $\diag_u(d_0,d_1)$. Then the map $\flat$ satisfies
\[
\flat V^2_1(v)=V^1_1,\quad \flat V^2_1(u)=V^2_1(u)
\]
where $V^1_1$ denotes the $1$-dimensional context of scalar matrices. 
The localization diagonalizes $\diag_v(e_0,e_1)$; we use conjugation by $v^*$ or $\sigma_1v^*$ as representatives of the two sampling outcomes. Beforehand, the prior probabilities are
\begin{align*}
\mu_\rho(X, V^2_1(u))
&=
\begin{cases}
0&X(V^2_1(u))=\emptyset\\ 
|u|_{\mathrm{top}(0)}^2&X(V^2_1(u))=\{\lambda_0^u\}\\
1-|u|_{\mathrm{top}(0)}^2&X(V^2_1(u))=\{\lambda_1^u\}\\
1&X(V^2_1(u))=\{\lambda_0^u,\lambda_1^u\}
\end{cases}\\
\mu_\rho(X, V^2_1(v))
&=
\begin{cases}
0&X(V^2_1(v))=\emptyset\\ 
|v|_{\mathrm{top}(0)}^2&X(V^2_1(v))=\{\lambda_0^v\}\\
1-|v|_{\mathrm{top}(0)}^2&X(V^2_1(v))=\{\lambda_1^v\}\\
1&X(V^2_1(v))=\{\lambda_0^v,\lambda_1^v\}
\end{cases}
\end{align*}
where $|\cdot|_{\mathrm{top}(\cdot)}^2$ denotes the squared absolute value of the top-left entry of a unitary matrix.
Now we perform the sampling.
\begin{enumerate}
\item[(i)] When $X(V^2_1(v))=\{\lambda_0^v\}$ becomes true with probability $|v|_{\mathrm{top}(0)}^2$, we conjugate the double section $D$ and its subsets $T$ and $S$ with $v^*$ to obtain the posterior probabilities of a proposition $Y$ as
\begin{align*}
\mu_\rho(Y, V^2_1(v^*u))
&=
\begin{cases}
0&Y(V^2_1(v^*u))=\emptyset\\ 
|v^*u|_{\mathrm{top}(0)}^2&Y(V^2_1(v^*u))=\{\lambda_0^{v^*u}\}\\
1-|v^*u|_{\mathrm{top}(0)}^2&Y(V^2_1(v^*u))=\{\lambda_1^{v^*u}\}\\
1&Y(V^2_1(v^*u))=\{\lambda_0^{v^*u},\lambda_1^{v^*u}\}
\end{cases}\\
\mu_\rho(Y, V^2_1(\mathrm{id}_2))
&=
\begin{cases}
0\qquad\qquad\qquad&Y(V^2_1(\mathrm{id}_2))\not\ni \lambda_0^{\mathrm{id}_2}\\
1&Y(V^2_1(\mathrm{id}_2))\ni\lambda_0^{\mathrm{id}_2}
\end{cases}.
\end{align*}
$Y$ is also a $\flat$-proposition if $Y(V^2_1(\mathrm{id}_2))$ is either $\emptyset$ or $\{\lambda_0^{\mathrm{id}_2}, \lambda_1^{\mathrm{id}_2}\}$. Although $D$ has no prominent role, in practice it prevents the swapping of events and counter-events within a fiber, while also contributing to the semantic correspondence between proposition $X$ and proposition $Y$. (Below, we omit such mechanisms.)

\item[(ii)] When $X(V^2_1(v))=\{\lambda_1^v\}$ becomes true with probability $1-|v|_{\mathrm{top}(0)}^2$, we take the conjugation with $\sigma_1v^*$ instead of $v^*$ to obtain 
\begin{align*}
\mu_\rho(Y, V^2_1(\sigma_1 v^*u))
&=
\begin{cases}
0&Y(V^2_1(\sigma_1v^*u))=\emptyset\\ 
|\sigma_1 v^*u|_{\mathrm{top}(0)}^2&Y(\textrm{do.})=\{\lambda_0^{\sigma_1 v^*u}\}\\
1-|\sigma_1 v^*u|_{\mathrm{top}(0)}^2&Y(\textrm{do.})=\{\lambda_1^{\sigma_1 v^*u}\}\\
1&Y(\textrm{do.})=\{\lambda_0^{\sigma_1 v^*u},\lambda_1^{\sigma_1 v^*u}\}
\end{cases}\\
\mu_\rho(Y, V^2_1(\mathrm{id}_2))
&=
\begin{cases}
0\qquad\qquad\qquad&Y(V^2_1(\mathrm{id}_2))\not\ni \lambda_0^{\mathrm{id}_2}\\
1&Y(V^2_1(\mathrm{id}_2))\ni\lambda_0^{\mathrm{id}_2}
\end{cases}
\end{align*}
where $|\sigma_1 v^*u|_{\mathrm{top}(0)}^2$ is equal to $1-|v^*u|_{\mathrm{top}(0)}^2$ though precisely the relevant propositions concern different contexts.
\end{enumerate} 
\label{1-qubit}
\end{Proposition}
\begin{Remark}
In conventional Bayesian statistics, updating probabilities is merely a matter of interpretation. However, our localization posits that observing data triggers an essential system change. Even so, putting $p=|v|_{\mathrm{top}(0)}^2$, $q=|u|_{\mathrm{top}(0)}^2$, and $P=|v^*u|_{\mathrm{top}(0)}^2$, we can interpret $P$ as a conditional probability once the sampling order is fixed. Let $Q$ be the probability of sampling the event with prior probability $p$ first. The sum of the probabilities $QpP+(1-Q)qP$ and $Qp(1-P)+(1-Q)(1-q)(1-P)$ is $p$ iff $Q=1$ for generic $u,v$. Thus the sampling order must be specified before entering the stochastic process. 
\end{Remark}
\begin{proof} The outer estimation, the daseinization, and the man are
\begin{align*}
(\mathrm{O},\delta)_{V^2_1(u)}(\pi)
&=
\begin{cases}
(0, \emptyset) & \pi =0\\ 
(\diag_u(1,0), \{\lambda_0^u\})& \pi=\diag_u(1,0)\\
(\diag_u(0,1), \{\lambda_1^u\})& \pi=\diag_u(0,1)\\
(\mathrm{id}_2, \{\lambda_0^u,\lambda_1^u\})& \textrm{otherwise}.
\end{cases}\\
\mu(X, V^2_1(u))
&=
\begin{cases}
0&X(V^2_1(u))=\emptyset\\ 
\diag_u(1,0)&X(V^2_1(u))=\{\lambda_0^u\}\\
\diag_u(0,1)&X(V^2_1(u))=\{\lambda_1^u\}\\
\mathrm{id}_2&X(V^2_1(u))=\{\lambda_0^u,\lambda_1^u\}.
\end{cases}
\end{align*}
Using real parameters we can write 
\[ 
u=\exp(i\varphi)\begin{pmatrix}
\cos\theta\exp(i\chi)& i\sin\theta \exp(i\psi) \\
i\sin\theta\exp(-i\psi)& \cos\theta \exp(-i\chi) 
\end{pmatrix}.
\]
Then a straightforward calculation shows
\[
\tr(\diag(1,0)\diag_u(1,0))=\cos^2\theta=|u|_{\mathrm{top}(0)}^2.
\]
We notice that $V_1^2(\sigma_1v^*u)$ is different from $V_1^2(v^*u)$ in general. 
\end{proof}

\subsection{$2$-qubit case}
This subsection discusses the case where $n=2$. 

First, we describe the minimal poset of contexts containing the Pauli matrices, which consists of the trivial context $V^1_2=\Cpx\textrm{id}_4$ and 
\begin{align*}
V^2_2(j,k)&=\langle\textrm{id}_4, \sigma_j\otimes \sigma_k\rangle&& (j,k) \neq (0,0),\\
V^4_2(j,k)&=\langle\textrm{id}_4, \sigma_j\otimes\sigma_0, \sigma_0\otimes\sigma_k, \sigma_j\otimes \sigma_k \rangle&&  j,k\neq 0,\\
V^4_2(j,k,l)&=\langle\textrm{id}_4, \sigma_j\otimes\sigma_1, \sigma_k\otimes\sigma_2, \sigma_l\otimes\sigma_3\rangle&& \{j,k,l\}=\{1,2,3\}
\end{align*}
where $\langle\cdot\rangle$ denotes the span. There are forty-five non-trivial inclusions
\begin{align*}
&V^2_2(j,0), V^2_2(0,k), V^2_2(j,k)\subseteq V^4_2(j,k) &&j,k\neq 0,\\
&V^2_2(j,1), V^2_2(k,2), V^2_2(l,3)\subseteq V^4_2(j,k,l)&&\{j,k,l\}=\{1,2,3\}
\end{align*}
This is the incident poset of the generalized quadrangle $\mathrm{GQ}(2,2)$; the $4$-dimensional contexts are the fifteen lines each including three points; the $2$-dimensional contexts are the fifteen points each on three lines. Note that, in a proper quadrangle, each side is defined by two of four vertices and each vertex is on two sides (dually to the quadrilateral).  
For nonzero $j$, choose $u(j)\in C_1/\mathrm{U}(1)$ with $u(j)\sigma_3u(j)^*=\sigma_j$ and put $u(j,k)=u(j)\otimes u(k)$. These give the nine contexts with product eigenbases. The stabilizer in $C_2/\mathrm{U}(1)$ of the diagonal poset
\[
\{V^1_2, V^2_2(0,3), V^2_2(3,0), V^2_2(3,3), V^4_2(3,3)\}
\] 
has order $768$. Hence $768\times(3^2-1)$ projective Clifford operators send its maximal context to one of the other eight product contexts. The identity $cV^4_2(2,2)c^*=V^4_2(2,3,1)$ illustrates the remaining six contexts, which have Bell eigenbases. We can choose $u(j,k,l)\in C_2$ such that 
\begin{align*}
V^2_2(j,1)&=u(j,k,l)V^2_2(3,0)u(j,k,l)^*,\\
V^2_2(k,2)&=u(j,k,l)V^2_2(0,3)u(j,k,l)^*,\\
V_2^2(l,3)&=u(j,k,l)V^2_2(3,3)u(j,k,l)^*,\\
V^4_2(j,k,l)&=u(j,k,l)V^4_2(3,3)u(j,k,l)^*
\end{align*} 
for each permutation $(j,k,l)$ of $(1,2,3)$. There are $768\times3!$ projective Clifford operators with these six target contexts. For any of the fifteen choices $u=u(j,k)$ or $u(j,k,l)$, we have
\[
uV^2_2(0,3)u^*, uV^2_2(3,0)u^*, uV^2_2(3,3)u^*\subseteq uV^4_2(3,3)u^*.
\]
Namely, the incident poset has the fifteen building blocks each of which is a unitary conjugation of the above diagonal poset.
 
Next, to consider a general choice of $T$ and $S$, we use the larger diagonal poset $\mathcal{P}\subset \V_A$ consisting of the fifteen elements
\begin{align*}
&V^1_2=\Cpx\mathrm{id}_4,\\
&V^4_2(3,3)=\{\diag(x_0,x_1,x_2,x_3)\mid x_0,x_1,x_2,x_3\in\Cpx\},\\ 
&V^2_2(3,0),~~ V^2_2(0,3),~~ V^2_2(3,3),\\
&V^2_2(0)=\{x_1=x_2=x_3\},~ V^2_2(1)=\{x_0=x_2=x_3\},\\
&V^2_2(2)=\{x_0=x_1=x_3\},~ V^2_2(3)=\{x_0=x_1=x_2\},\\
&V^3_2(0,2)=\{x_0=x_2\}, V^3_2(0,1)=\{x_0=x_1\}, V^3_2(0,3)=\{x_0=x_3\},\\
&V^3_2(1,3)=\{x_1=x_3\}, V^3_2(2,3)=\{x_2=x_3\}, V^3_2(1,2)=\{x_1=x_2\}.
\end{align*}
Each building block of $\V$ is a unitary conjugation of a subposet of $\mathcal{P}$. 

In the case where $\dim H_2\geq 3$, the localizability implies that $\dim H_2=4$ and the poset $\overline{\flat(\V)}$ includes a conjugation of $V_2^4(3,3)$. 

\begin{Proposition} In the $2$-qubit updating, suppose that $\overline{\flat(\V)}$ contains $vV_2^4(3,3)v^*$ for a $2$-qubit unitary matrix $v\in \mathrm{U}(4)$. Any building block of $\flat(\V)$ is a subset of $u\mathcal{P}u^*$ for some $u\in\mathrm{U}(4)$ and has the probabilities
\begin{align*}
&\tr(\rho\diag_u(1,0,0,0))=|u|_{\mathrm{top}(0)}^2,~~
\tr(\rho\diag_u(0,1,0,0))=|u|_{\mathrm{top}(1)}^2,\\
&\tr(\rho\diag_u(0,0,1,0))=|u|_{\mathrm{top}(2)}^2,~~
\tr(\rho\diag_u(0,0,0,1))=|u|_{\mathrm{top}(3)}^2\\
&(\textrm{and therefore $\tr(\rho\diag_u(1,0,1,0))=|u|_{\mathrm{top}(0)}^2+|u|_{\mathrm{top}(2)}^2$ etc.})
\end{align*}
which are updated  to 
$|v^*u|_{\mathrm{top}(0)}^2,\dots,$
$|v^*u|_{\mathrm{top}(3)}^2$
respectively under the data with probability $|v|_{\mathrm{top}(0)}^2$ where $|\cdot|_{\mathrm{top}(j)}^2$ denotes the squared absolute value of the $j$-component of the top row of a unitary matrix. 
\end{Proposition}

\begin{proof}
Although the proposition does not specify for which proposition and in what context the probability applies, we clarify this point. 
For a proposition $X$ and $V\in u\mathcal{P}u^*\cap\V$, we represent $X(V)$ as a subset of 
\[
\Sigma_A(V_2^4(u))=\left\{\begin{array}{l}
\lambda^u_{0}: \diag_{u}(x_0,x_1,x_2,x_3)\mapsto x_0,\\
\lambda^u_{1}: \diag_{u}(x_0,x_1,x_2,x_3)\mapsto x_1,\\
\lambda^u_{2}: \diag_{u}(x_0,x_1,x_2,x_3)\mapsto x_2,\\
\lambda^u_{3}: \diag_{u}(x_0,x_1,x_2,x_3)\mapsto x_3
\end{array}\right\}.
\]
The man $\mu(X,V)$ is then the projection $\diag_u(\chi)$ for the characteristic function $\chi$ of the subset $X(V)$. For example, we have
\[
\mu(\{\lambda_0^u,\lambda_2^u\},V_2^4(u))=\diag_u(1,0,1,0)~ (\chi
=\begin{cases}
1 &\lambda_j^u\in  \{\lambda_0^u,\lambda_2^u\}\\
0 &\lambda_j^u\not\in \{\lambda_0^u,\lambda_2^u\}
\end{cases}).
\]
Conjugation by $v^*$ gives $\mu(Y,V)=\diag_{v^*u}(\chi)$ for $V\in v^*u\mathcal{P}u^*v\cap v^*\V v$, where $\chi$ is the characteristic function of the inverse image representing $Y(V)$.
Excluding the sampled context $V_2^4(\mathrm{id}_4)$ from prediction means that $\flat$-propositions distinguish no individual sampling outcomes: their restrictions to the subposet below it are either empty or whole.
\end{proof}

\begin{Example}
Take the set $T=\{\sigma_j\otimes\sigma_k\mid 0\leq j,k\leq3\}$ of Pauli representatives and put $S=T\setminus\{\sigma_0\otimes \sigma_1, \sigma_1\otimes \sigma_0, \sigma_1\otimes \sigma_1\}$. The set $T$ yields the same context poset as $P_2$. Then the complementary poset is 
\begin{align*}
\overline{\flat(\V)}
&=\{V_2^1, V_2^2(0,1), V_2^2(1,0), V_2^2(1,1), V_2^4(1,1)\}\\
&=h\otimes h\{V_2^1, V_2^2(0,3), V_2^2(3,0), V_2^2(3,3), V_2^4(3,3)\}h\otimes h
\end{align*}
and we represent the four sampling outcomes by $v^*=h\otimes h, h\otimes(\sigma_1h), (\sigma_1h)\otimes h, (\sigma_1h)\otimes(\sigma_1h)$ 
with probabilities $\displaystyle |v|_{\mathrm{top}(j)}^2=\frac{1}{4},~ \frac{1}{4},~ \frac{1}{4},~ \frac{1}{4}$, respectively. 
\begin{enumerate}

\item[(i)] 
Take $u=h\otimes sh$ as an example and consider the subposet 
\begin{align*}
u\{V_2^1, V_2^2(0,3), V_2^2(3,3), V_2^4(3,3)\}u^*&\\
=\{V_2^1, V_2^2(0,2), V_2^2(1,2), V_2^4(1,2)\}&\subseteq \flat(\V)
\end{align*}
with the prior probability distribution
\[
|u|_{\mathrm{top}(j)}^2=\frac{1}{4},~ \frac{1}{4},~ \frac{1}{4},~ \frac{1}{4} \quad (j=0,1,2,3).
\]
If $v^*=h\otimes h$ is chosen with probability $\displaystyle \frac{1}{4}$, the subposet becomes
\begin{align*}
v^*u\{V_2^1, V_2^2(0,3), V_2^2(3,3), V^4_2(3,3)\}u^*v&\\
=\{V_2^1, V_2^2(0,2), V_2^2(3,2), V^4_2(3,2)\}&\subseteq v^*\flat(\V)v
\end{align*}
with the posterior probability distribution
\[
|v^*u|_{\mathrm{top}(j)}^2=\frac{1}{2},~ \frac{1}{2},~ 0,~ 0\quad (j=0,1,2,3)
\]
because $\displaystyle v^*u=\mathrm{id}_2\otimes \frac{1}{\sqrt{2}}\begin{pmatrix}
\exp(i\pi/4)& \exp(-i\pi/4)\\
\exp(-i\pi/4)&\exp(i\pi/4)
\end{pmatrix}$. 

\item[(ii)] 
Take $u=c(h\otimes \mathrm{id}_2)$ as another example and consider 
\begin{align*}
&u\{V_2^1, V_2^2(0,3), V_2^2(3,3), V_2^4(3,3)\}u^*\\
&=\{V_2^1, V_2^2(2,2), V_2^2(3,3), V_2^4(1,2,3)\}\subseteq \flat(\V)
\end{align*}
with the prior probability distribution
\[
|u|_{\mathrm{top}(j)}^2=\frac{1}{2},~ 0,~ \frac{1}{2},~ 0 \quad (j=0,1,2,3).
\]
If $v^*=h\otimes h$ with probability $\displaystyle\frac{1}{4}$, we have
\begin{align*}
&v^*u\{V_2^1, V_2^2(0,3), V_2^2(3,3), V^4_2(3,3)\}u^*v\\
&=\{V_2^1, V_2^2(1,1), V_2^2(2,2), V_2^4(1,2,3)\}\subseteq v^*\flat(\V)v
\end{align*}
with the posterior probability distribution
\[
|v^*u|_{\mathrm{top}(j)}^2=\frac{1}{2},~\frac{1}{2},~ 0,~ 0, \quad (j=0,1,2,3).
\]

\end{enumerate}
Regardless of the sampling result, there is no nontrivial diagonal context in the new poset $v^*\flat(\V)v$.
\label{hh}
\end{Example}

In the case where $\dim H_2=2$, the localizability implies that the poset $\overline{\flat(\V)}$ is a conjugation of the subposet $\{V^1_2,~ V^2_2(0,3)\}\subseteq \mathcal{P}$.  

\begin{Proposition} In the $2$-qubit updating, suppose $\dim H_2=2$ and 
\[
\overline{\flat(\V)}=\{V^1_2, V^1_1\otimes V^2_1(v)\}
\]
for a $1$-qubit unitary matrix $v\in \mathrm{U}(2)$. 
For a building block of $\flat(\V)$ included in $u\mathcal{P}u^*$, the probabilities $|u|_{\mathrm{top}(0)}^2$, $|u|_{\mathrm{top}(1)}^2$, $|u|_{\mathrm{top}(2)}^2$, $|u|_{\mathrm{top}(3)}^2$ are updated to 
$|\widetilde{v}^*u|_{\mathrm{top}(0)}^2$,
$|\widetilde{v}^*u|_{\mathrm{top}(1)}^2$,
$|\widetilde{v}^*u|_{\mathrm{top}(2)}^2$,
$|\widetilde{v}^*u|_{\mathrm{top}(3)}^2$.
If the matrix $u$ splits as $u_1\otimes u_2$, then these probabilities are 
\begin{align*}
&|u_1|_{\mathrm{top}(0)}^2|u_2|_{\mathrm{top}(0)}^2,&&
|u_1|_{\mathrm{top}(0)}^2(1-|u_2|_{\mathrm{top}(0)}^2),\\
&(1-|u_1|_{\mathrm{top}(0)}^2)|u_2|_{\mathrm{top}(0)}^2,&&
(1-|u_1|_{\mathrm{top}(0)}^2)(1-|u_2|_{\mathrm{top}(0)}^2)
\end{align*}
and updated to 
\begin{align*}
&|u_1|_{\mathrm{top}(0)}^2|v^*u_2|_{\mathrm{top}(0)}^2,&&
|u_1|_{\mathrm{top}(0)}^2(1-|v^*u_2|_{\mathrm{top}(0)}^2),\\
&(1-|u_1|_{\mathrm{top}(0)}^2)|v^*u_2|_{\mathrm{top}(0)}^2,&&
(1-|u_1|_{\mathrm{top}(0)}^2)(1-|v^*u_2|_{\mathrm{top}(0)}^2).
\end{align*}
\label{2dim}
\end{Proposition}

\begin{proof}
For a proposition $X$ and a context $V$ in the chosen block, we represent $X(V)$ as a subset of 
\[
\Sigma_A(V_2^4(u))=\left\{\begin{array}{l}
\lambda^u_{0}: \diag_{u}(x_0,x_1,x_2,x_3)\mapsto x_0,\\
\lambda^u_{1}: \diag_{u}(x_0,x_1,x_2,x_3)\mapsto x_1,\\
\lambda^u_{2}: \diag_{u}(x_0,x_1,x_2,x_3)\mapsto x_2,\\
\lambda^u_{3}: \diag_{u}(x_0,x_1,x_2,x_3)\mapsto x_3
\end{array}\right\}
\]
and on the complementary poset we use 
\[
\Sigma(V^1_1\otimes V^2_1(v))=\left\{\begin{array}{l}
\lambda^v_{\bullet 0}: \mathrm{id}\otimes \diag_{v}(x_0,x_1)\mapsto x_0,\\
\lambda^v_{\bullet 1}: \mathrm{id}\otimes \diag_{v}(x_0,x_1)\mapsto x_1\\
\end{array}\right\}.
\]
In the former case, the man $\mu(X,V)$ is the projection $\diag_u(\chi)$; in the latter case, $\mu(X,V)=\mathrm{id}_2\otimes \diag_v(\chi)$. In each case, the prior probability is the sum of the squared absolute values of the top entries corresponding to the support of $\chi$. After conjugation by $\widetilde{v}^*$, the same construction gives $\mu(Y,V)=\diag_{\widetilde{v}^*u}(\chi)$ on the transported block. The product formulas follow from Proposition \ref{1-qubit}.
Excluding the sampled context $V^1_1\otimes V^2_1(\mathrm{id}_2)$ from prediction means that $\flat$-propositions distinguish no individual sampling outcomes: their restrictions to the subposet below it are either empty or whole. 
\end{proof}

\begin{Remark}
While it is difficult to calculate probabilities in a general form for the non-split case, the calculations for each individual case are straightforward. In quantum computing, the probability for the non-split or entangled case is of particular importance.
\end{Remark}

\begin{Example}
Take the same set $T=\{\sigma_j\otimes\sigma_k\mid 0\leq j,k\leq3\}$ of Pauli representatives and put $S=T\setminus\{\sigma_0\otimes \sigma_1\}$. Then 
\[
\overline{\flat(\V)}=\{V_2^1, V_2^2(0,1)\}=\mathrm{id}_2\otimes h\{V_2^1, V_2^2(0,3)\}\mathrm{id}_2\otimes h
\]
and we represent the two sampling outcomes by $v^*=h$ or $\sigma_1h$ 
with probabilities $\displaystyle |v|_{\mathrm{top}(j)}^2=\frac{1}{2},~\frac{1}{2}$. 

\begin{enumerate}

\item[(i)] 
Take $u=h\otimes sh$ and consider again the subposet 
\[
\{V_2^1, V_2^2(0,2), V_2^2(1,2), V_2^4(1,2)\}\subseteq \flat(\V)
\] 
with the prior probability distribution
\[
|u|_{\mathrm{top}(j)}^2=\frac{1}{4},~ \frac{1}{4},~ \frac{1}{4},~ \frac{1}{4} \quad (j=0,1,2,3).
\]
If $v^*=h$ is chosen with probability $\displaystyle\frac{1}{2}$, the subposet and 
the probability are preserved ($|\widetilde{v}^*u|_{\mathrm{top}(j)}^2=|u|_{\mathrm{top}(j)}^2$). 

\item[(ii)] 
Take $u=c(h\otimes \mathrm{id}_2)$ and consider again 
\[
\{V_2^1, V_2^2(1,1), V_2^2(2,2), V_2^2(3,3), V_2^4(1,2,3)\}\subseteq \flat(\V)
\]
with the prior probability distribution
\[
|u|_{\mathrm{top}(j)}^2=\frac{1}{2},~ 0,~ \frac{1}{2},~ 0 \quad (j=0,1,2,3).
\]
If $v^*=h$ is chosen with probability $\displaystyle\frac{1}{2}$, the subposet becomes 
\[
\{V_2^1, V_2^2(1,3), V_2^2(2,2), V_2^2(3,1), V_2^4(3,2,1)\}\subseteq \widetilde{v}^*\flat(\V)\widetilde{v}
\]
with the posterior probability distribution
\[
|\widetilde{v}^*u|_{\mathrm{top}(j)}^2=\frac{1}{4},~ \frac{1}{4},~ \frac{1}{4},~ \frac{1}{4} \quad (j=0,1,2,3).
\]

\end{enumerate}
For any sampling results, there remains a diagonal context in the new poset, indicating that Example \ref{hh} shows simultaneous observation.
\label{ih}
\end{Example}

\section{Conclusions}
Quantum theory has traditionally treated the change associated with measurement as a sudden jump within the theory. In this paper, we have formalized the change as a dynamic shift in the logical boundaries of the theory itself. Through the integration of Nakayama's context-selection and the idea of Bayesian updating, we described how a predictable variable becomes a given constant through a sampling.

The arrow of time here refers to the exclusion of the determined context from subsequent prediction.  We calculated some 2-qubit gates to illustrate how the framework can track the structural evolution in an actual computation. A related contact/symplectic topological approach to quantum contexts and the choice of subsystems, using planar open-books, is developed in \cite{MoriPlanar}.


\begin{thebibliography}{99}
\bibitem{DI1} A.~D\"oring and C.J.~Isham, ``A topos foundation for theoretical physics: I. Formal language for physics'', J. Math. Phys. 49, 053515 (2008). https://doi.org/10.1063/1.2883740
\bibitem{DI2} A.~D\"oring and C.J.~Isham, ``A topos foundation for theoretical physics: II. Daseinization and the liberation of quantum theory,'' J. Math. Phys. 49, 053516 (2008). https://doi.org/10.1063/1.2883742
\bibitem{DI3} A.~D\"oring and C.J.~Isham, ``A topos foundation for theoretical physics: III. The representation of physical quantities with arrows,'' J. Math. Phys. 49, 053517 (2008). https://doi.org/10.1063/1.2883777
\bibitem{DI4} A.~D\"oring and C.J.~Isham, ``A topos foundation for theoretical physics: IV. Categories of systems'', J. Math. Phys. 49, 053518 (2008). https://doi.org/10.1063/1.2883826
\bibitem{D5} A.~D\"oring, ``Quantum states and measures on the spectral presheaf'', Adv. Sci. Lett. 2, 291--301 (2009). https://doi.org/10.1166/asl.2009.1037
\bibitem{DI6} A.~D\"oring and C.J~Isham, ``Classical and quantum probabilities as truth values'', J. Math. Phys. 53, 032101 (2012). https://doi.org/10.1063/1.3688627
\bibitem{Nakayama1} K.~Nakayama, ``Topologies on quantum topoi induced by quantization'', J. Math. Phys. 54, 072102 (2013). https://doi.org/10.1063/1.4813960
\bibitem{Nakayama2} K.~Nakayama, ``Topos quantum theory on quantization induced sheaves'', J. Math. Phys. 55, 102103 (2014). https://doi.org/10.1063/1.4898185
\bibitem{Nakayama3} K.~Nakayama, ``Topos quantum theory reduced by context-selection functors'', J. Math. Phys. 57, 122103 (2016).  https://doi.org/10.1063/1.4972215

\bibitem{Critique} A.~Khrennikov, ``Ozawa's intersubjectivity theorem as objection to QBism individual agent perspective'', Int. J. Theor. Phys. 63, 23 (2024). https://doi.org/10.1007/s10773-024-05552-8

\bibitem{NC} M.A.~Nielsen, I.L.~Chuang, ``Quantum computation and quantum information: 10th anniversary edition. (Cambridge: Cambridge University Press, 2010) 

\bibitem{Got} D.~Gottesman, ``The Heisenberg representation of quantum computers'', Group22: Proceedings of the XXII International Colloquium on Group Theoretical Methods in Physics, eds. S. P. Corney, R. Delbourgo, and P. D. Jarvis, pp. 32--43 (Cambridge, International Press, 1999)

\bibitem{Boykin} P.O.~Boykin, T.~Mor, M.~Pulver, V.P.~Roychowdhury, F.~Vatan, ``On universal and fault-tolerant quantum computing: a novel basis and a new constructive proof of universality for Shor's basis'', Proc. 40th FOCS, pp. 486--494, 1999.

\bibitem{Takesaki} M.~Takesaki, ``Theory of Operator Algebras'', vol. I, II, III, Springer. Encyclopaedia of Mathematical Sciences 124 (2002), 125, 127 (2003). 

\bibitem{MoriPlanar} A.~Mori, ``Planar contact structures with Calabi-Yau fillings and topological quantum computation'', preprint (2026), arXiv:2609.30406. 

\end{thebibliography}
\end{document}